\pdfoutput=1
\IfFileExists{ieeeconf.cls}{%
  \documentclass[letterpaper,10pt,conference]{ieeeconf}%
}{%
  \documentclass[letterpaper,10pt,conference]{IEEEtran}%
  \providecommand{\overrideIEEEmargins}{}%
}
\IEEEoverridecommandlockouts
\usepackage{amsmath,amssymb,amsfonts}
\usepackage{mathtools}
\usepackage{bm}
\usepackage{booktabs}
\usepackage{array}
\usepackage{cite}
\usepackage{url}
\usepackage{graphicx}
\usepackage{balance}
\usepackage{tikz}
\usetikzlibrary{arrows.meta,positioning,calc,shadings,decorations.pathmorphing}

\makeatletter
\def\@begintheorem#1#2{\trivlist\item[\hskip\labelsep{\bfseries #1\ #2:}]\itshape}
\def\@opargbegintheorem#1#2#3{\trivlist\item[\hskip\labelsep{\bfseries #1\ #2\ (#3):}]\itshape}
\makeatother
\newtheorem{proposition}{Proposition}

\newcommand{\R}{\mathbb{R}}

\newcommand{\I}{\mathcal{I}}
\newcommand{\J}{\mathcal{J}}
\newcommand{\F}{\mathcal{F}}
\newcommand{\N}{\mathcal{N}}
\newcommand{\tr}{\operatorname{tr}}

\newcommand{\KL}{\operatorname{KL}}
\newcommand{\Tcone}{T_{\Theta}}
\newcommand{\Ncone}{N_{\Theta}}
\newcommand{\ind}{\iota_{\Theta}}

\title{Adaptive Control as an Information Gradient Flow:\\
Excitation, Constraints, Energy, and Stochastic Learning}

\author{Omkar Sudhir Patil%
\thanks{O. S. Patil is with the Department of Electrical and Computer Engineering, Louisiana State University, Baton Rouge, LA, USA.}%
}

\begin{document}
\maketitle
\thispagestyle{empty}
\pagestyle{empty}

\begin{abstract}
Adaptive control is usually developed through Lyapunov stability, excitation, and parameter convergence, whereas neighboring fields describe learning using convexity, information geometry, variational principles, and stochastic thermodynamics. This paper develops a constrained common language for these viewpoints. For linearly parameterized models, the excitation Gramian is simultaneously the Hessian of accumulated prediction loss and, under Gaussian observations, Fisher information up to scaling. Persistent, finite, and partial excitation become uniform, finite-horizon, and restricted temporal curvature; memory methods retain previously acquired curvature. On a closed convex parameter set, deterministic adaptation is a projected gradient/Onsager flow, with tangent and normal cones recovering projection, while concurrent learning supplies data-dependent symmetric dissipation. Reflected Langevin dynamics, no-flux Fokker--Planck evolution, and constrained free-energy flow provide the stochastic counterpart. The framework distinguishes information supplied by data from confinement supplied by regularization or hard constraints and connects familiar adaptive-control structures to modern variational and information-theoretic language.
\end{abstract}

\section{Introduction}
Adaptive control combines closed-loop stability with online parameter estimation. Lyapunov analysis can establish boundedness and tracking without parameter convergence, which depends on information supplied by the observed trajectory. Persistent excitation (PE) is the classical condition for convergence; finite- and partial-excitation conditions and memory methods relax sustained excitation \cite{anderson1982,boyd1983,green1986,chowdhary2010,chowdhary2011,parikh2019}. Statistics, optimization, and physics describe learning through loss curvature, Fisher information, gradient flows, dissipation, and free energy \cite{amari1998,amari2016,ly2017,jordan1998}. Gaudio et al.\ connect adaptive control to machine-learning optimization through update laws, regularization, projection, adaptive gains, and higher-order learning \cite{gaudio2019}. Building on these connections, a remaining question is how to develop a common description of accumulated data information, constrained parameter motion, closed-loop dissipation, and stochastic free energy.

More specifically, three questions motivate the present work. First, how can persistent, finite, partial, and retained excitation be represented through the curvature and statistical information of a single accumulated loss? Second, how do adaptive gains and parameter constraints define the geometry of learning, and how does information-driven parameter dissipation interact with the Lyapunov function used for closed-loop stability? Third, what stochastic formulation preserves this information--dissipation structure when parameter motion includes diffusion and boundary reflection? Addressing these questions requires distinguishing data information from confinement imposed by regularization or hard constraints. Bounded estimates and small tracking errors alone do not establish parameter identifiability, particularly for large models learned from limited data.

This paper develops a common mathematical framework based on accumulated prediction loss over a closed convex parameter set. For linearly parameterized models, the excitation Gramian is the Hessian of this loss and, under Gaussian observations, the Fisher information matrix up to noise scaling. Persistent, finite, and partial excitation correspond to uniform, finite-horizon, and restricted temporal curvature, respectively, while memory methods retain previously acquired curvature. Adaptive gains define the metric of a projected gradient or Onsager flow, and constraints enter through tangent cones and normal-cone reactions. A decomposition of the closed-loop error dynamics separates the Lyapunov function from the information energy and identifies the parameter-space dissipation supplied by concurrent learning. Reflected Langevin dynamics and their Fokker--Planck equation extend this formulation to stochastic learning through constrained free-energy dissipation. The framework thus distinguishes information acquired from data from confinement imposed by regularization or feasibility constraints and extends the discussion to nonlinear parameterizations through their Jacobian Gramians.

\begin{figure*}[t]
\centering
\resizebox{\textwidth}{!}{\input{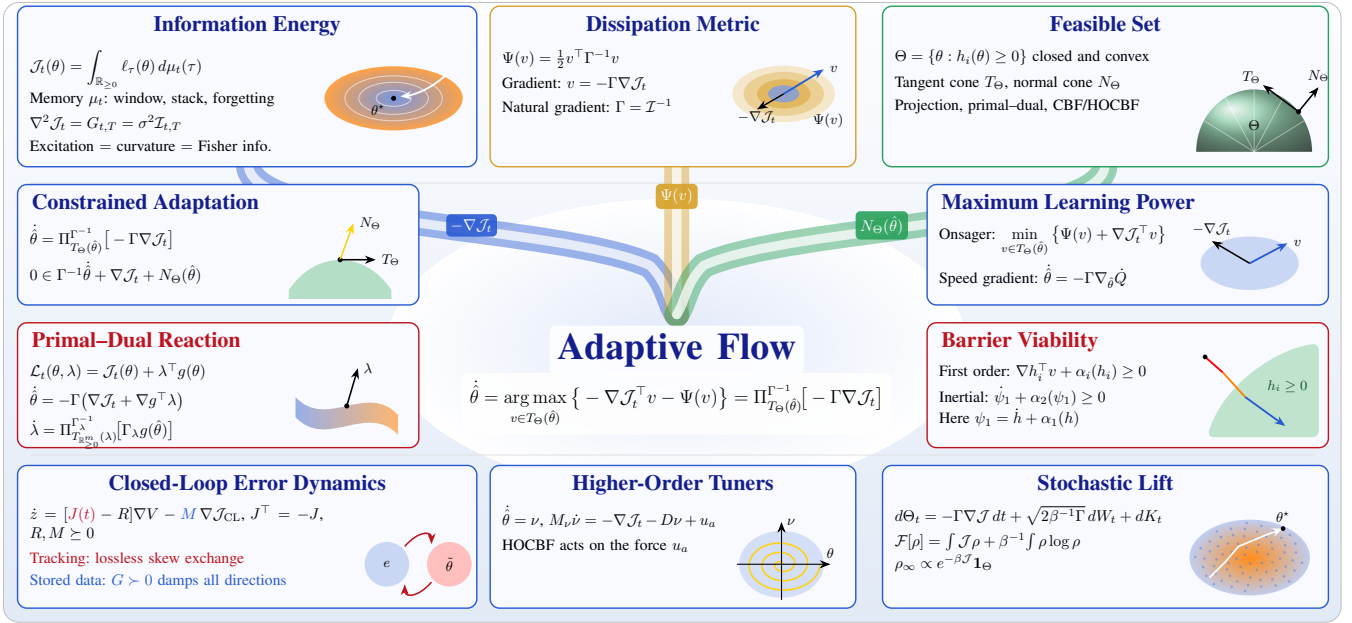}}
\caption{Overview. The information force, the dissipation metric, and the constraint reaction merge in the adaptive flow (center); the surrounding panels are its constrained, closed-loop, inertial, and stochastic realizations.}
\label{fig:overview}
\end{figure*}
Fig.~\ref{fig:overview} summarizes these correspondences. Section~\ref{sec:energy} relates accumulated-loss curvature to excitation, and Section~\ref{sec:infogeo} connects this curvature to Fisher information and adaptive gains to a metric. Section~\ref{sec:lyap} develops the Lyapunov and Onsager formulations, Section~\ref{sec:stoch} presents the stochastic and free-energy counterparts, and Section~\ref{sec:unified} discusses nonlinear parameterizations. Section~\ref{sec:conclusion} concludes.

\section{Information Energy and Temporal Convexification}\label{sec:energy}
\subsection{Instantaneous and accumulated losses}\label{sec:losses}
Let $\Theta\subset\R^p$ be a nonempty closed convex set, let $\theta^\star\in\Theta$ be the unknown parameter, and let $\phi:\R_{\geq0}\to\R^p$ and $y:\R_{\geq0}\to\R$ be measurable, locally bounded signals related by the regression model
\begin{equation}
    y(t)=\phi(t)^\top\theta^\star+\varepsilon(t),
    \label{eq:regression}
\end{equation}
where $\varepsilon:\R_{\geq0}\to\R$ is a measurement disturbance. Begin with $\varepsilon\equiv0$ and define the instantaneous loss $\ell_t:\R^p\to\R_{\geq0}$ by
\begin{equation}
    \ell_t(\theta)\triangleq\frac{1}{2}\big(\phi(t)^\top\theta-y(t)\big)^2.
    \label{eq:instantloss}
\end{equation}
Its Hessian is
\begin{equation}
    \nabla^2\ell_t(\theta)=\phi(t)\phi(t)^\top.
    \label{eq:rankone}
\end{equation}
For $p>1$, a single observation therefore supplies curvature in at most one parameter direction, which is the optimization counterpart of the fact that one scalar equation cannot identify a parameter vector.

Let $\mu_t$ be a finite nonnegative Borel measure on $\R_{\geq0}$ that encodes how the data available at time $t$ are retained, and let $\ind:\R^p\to\{0,\infty\}$ be the convex indicator of $\Theta$, equal to $0$ on $\Theta$ and $\infty$ elsewhere. The accumulated loss $\J_{\mu,t}^{\Theta}:\R^p\to\R\cup\{\infty\}$ is
\begin{equation}
    \J_{\mu,t}^{\Theta}(\theta)\triangleq\int_{\R_{\geq0}} \ell_\tau(\theta)\,d\mu_t(\tau)+\ind(\theta).
    \label{eq:generalJ}
\end{equation}
Its smooth part is $\J_t\triangleq\int_{\R_{\geq0}}\ell_\tau\,d\mu_t(\tau)$, so that $\J_{\mu,t}^{\Theta}=\J_t+\ind$; the notation $\J_t$ is used whenever the choice of measure is immaterial. For a closed convex set $C\subset\R^p$ and $\theta\in C$, let $T_C(\theta)\triangleq\operatorname{cl}\{v\in\R^p:\theta+hv\in C\text{ for some }h>0\}$ denote the tangent cone and $N_C(\theta)\triangleq\{n\in\R^p:n^\top(\vartheta-\theta)\leq0\ \forall\vartheta\in C\}$ the normal cone. The indicator contributes no data curvature in the interior of $\Theta$; at $\theta\in\partial\Theta$ it restricts feasible displacements to $\Tcone(\theta)$. Throughout the paper, curvature refers to data curvature along admissible directions and never to curvature induced by the boundary.

For the Lebesgue measure on a window $[t,t+T]$, \eqref{eq:generalJ} reduces to the window loss
\begin{equation}
    \J_{t,T}(\theta)\triangleq\int_t^{t+T}\ell_\tau(\theta)d\tau.
    \label{eq:windowloss}
\end{equation}
Then
\begin{equation}
    \nabla^2\J_{t,T}(\theta)
    =G_{t,T}
    \triangleq\int_t^{t+T}\phi(\tau)\phi(\tau)^\top d\tau.
    \label{eq:gramian}
\end{equation}
Thus, although every instantaneous Hessian can be rank deficient, temporal aggregation can generate a full-rank Hessian. Temporal aggregation therefore convexifies the identification problem even when no single observation does. The following proposition states the resulting correspondence.

\begin{proposition}
For \eqref{eq:regression} with $\varepsilon\equiv0$, the regressor $\phi$ is persistently exciting with constants $T>0$ and $\alpha>0$, i.e.,
\begin{equation}
    G_{t,T}\succeq \alpha I,\qquad \forall t\geq 0,
    \label{eq:PE}
\end{equation}
if and only if the family of window losses $\{\J_{t,T}\}_{t\geq0}$ is uniformly $\alpha$-strongly convex on $\R^p$.
\end{proposition}
\emph{Proof:} From \eqref{eq:gramian}, $\nabla^2\J_{t,T}=G_{t,T}$, and a twice differentiable function is $\alpha$-strongly convex if and only if its Hessian is bounded below by $\alpha I$. \hfill$\square$

When PE is stated with the customary upper Gramian bound $G_{t,T}\preceq\bar\alpha I$, the same family is also uniformly $\bar\alpha$-smooth; for bounded $\phi$ and fixed $T$, such an upper bound is automatic. On $\Theta$, \eqref{eq:PE} implies strong convexity of $\J_{t,T}+\ind$, and in general only curvature along feasible directions matters. The proposition reinterprets PE as a condition on the loss rather than on the signal: every moving window yields an $\alpha$-strongly convex accumulated loss. Classical results linking PE to exponential parameter convergence \cite{anderson1982,boyd1983,green1986} can therefore be read as convergence results under uniform temporal curvature. Section~\ref{sec:relaxed} relaxes the uniformity to finite, partial, and retained curvature.

\subsection{Finite, partial, and retained curvature}\label{sec:relaxed}
The same viewpoint distinguishes several relaxed excitation notions. If there exists a finite interval $[t_0,t_0+T]$ for which
\begin{equation}
    G_{t_0,T}\succ0,
    \label{eq:FE}
\end{equation}
then the corresponding finite-horizon information energy is strongly convex even if later moving windows are not. Finite-excitation and memory-based approaches \cite{parikh2019,cho2018} exploit exactly this finite-horizon strong convexity, and initial-excitation methods impose \eqref{eq:FE} through a continuous-time integral over an initial interval \cite{roy2018}.

If $G_{t,T}$ has rank $r<p$, the energy is flat on $\ker G_{t,T}$ but curved on $\mathcal{S}_{t,T}=\operatorname{range}G_{t,T}$. For $v\in\mathcal{S}_{t,T}$,
\begin{equation}
    v^\top G_{t,T}v\geq
    \lambda_{\min}^{+}(G_{t,T})\|v\|^2,
    \label{eq:restricted}
\end{equation}
where $\lambda_{\min}^{+}$ is the smallest positive eigenvalue. Partial excitation is therefore naturally interpreted as restricted strong convexity. With constraints, a null direction that points outside $\Tcone(\hat\theta)$ is inadmissible even though it remains uninformative, so the data identify an equivalence set rather than an isolated parameter,
\begin{equation}
    \Theta_{t,T}^\star
    =(\theta^\star+\ker G_{t,T})\cap\Theta.
    \label{eq:minset}
\end{equation}

Historical-data methods can be represented by choosing a nonlocal measure in \eqref{eq:generalJ}. For a history stack of sampling times $\{t_i\}_{i=1}^N\subset[0,t]$ with weights $w_i>0$ and Dirac measures $\delta_{t_i}$,
\begin{equation}
    d\mu_t(\tau)=\sum_{i=1}^N w_i\delta_{t_i}(d\tau),
\end{equation}
which yields
\begin{equation}
    \nabla^2\J_{\mu,t}
    =\sum_{i=1}^Nw_i\phi(t_i)\phi(t_i)^\top.
    \label{eq:CLgramian}
\end{equation}
The rank condition used in concurrent learning can therefore be read as strong convexity of a discrete empirical information energy. Concurrent learning does not require the trajectory to continuously regenerate curvature; it retains curvature that was acquired earlier \cite{chowdhary2010,chowdhary2011}. Integral concurrent learning similarly constructs usable historical information without direct state-derivative measurements and establishes parameter convergence under a finite-excitation condition \cite{parikh2019}.

For exponential forgetting with rate $\gamma_f>0$,
\begin{equation}
    d\mu_t(\tau)=e^{-\gamma_f(t-\tau)}\mathbf{1}_{[0,t]}(\tau)d\tau,
\end{equation}
where $\mathbf{1}_{[0,t]}$ is the indicator of $[0,t]$, old curvature continuously decays. Classical recursive least-squares results with forgetting establish exponential convergence under PE \cite{johnstone1982}. The measure viewpoint makes the distinction between acquiring, retaining, and forgetting information explicit. Section~\ref{sec:infogeo} reads the same curvature as Fisher information.

\section{From Curvature to Information Geometry}\label{sec:infogeo}
\subsection{The same Gramian as Fisher information}\label{sec:fisher}
Now suppose the observations satisfy
\begin{equation}
    dy_t=\phi(t)^\top\theta^\star dt+\sigma d\omega_t,
    \label{eq:stochobs}
\end{equation}
where $\omega_t$ is a scalar standard Wiener process and $\sigma>0$ is known. By Girsanov's theorem, the negative log-likelihood of $\theta$ over $[t,t+T]$ is quadratic in $\theta$ up to a $\theta$-independent constant, and its Hessian, the Fisher information matrix $\I_{t,T}(\theta)\in\R^{p\times p}$, is
\begin{equation}
    \I_{t,T}(\theta)
    \triangleq\frac{1}{\sigma^2}
    \int_t^{t+T}\phi(\tau)\phi(\tau)^\top d\tau
    =\frac{1}{\sigma^2}G_{t,T}.
    \label{eq:fisher}
\end{equation}
Hence the excitation Gramian, the Hessian of the accumulated loss, and the Fisher information matrix \cite{ly2017} coincide up to noise scaling. Equation \eqref{eq:fisher} therefore gives an information-theoretic reading of excitation: PE requires every parameter direction to receive a uniformly positive amount of statistical distinguishability over each window. If $\theta^\star$ lies on $\partial\Theta$, only the Fisher quadratic form restricted to $\Tcone(\theta^\star)$ matters: a constraint removes alternatives, but it does not create Fisher information in directions the experiment never probes.

The Fisher reading also explains why different scalar summaries of $G$ encode different design objectives. The quantity $\lambda_{\min}(G)$ is the E-optimality criterion of experiment design, $\log\det G$ is the D-optimality criterion, and $\tr(G^{-1})$, when $G\succ0$, is the A-optimality criterion \cite{pukelsheim2006}: the first measures the worst-informed direction, the second the information volume, and the third the aggregate estimation variance. The adaptive-control emphasis on $\lambda_{\min}$ reflects its need for uniform convergence guarantees in all parameter directions. Section~\ref{sec:metric} turns from the Hessian to the metric in which the estimate moves.

\subsection{Metric interpretation of adaptive gains}\label{sec:metric}
Fisher information also defines a Riemannian metric on a statistical parameter manifold \cite{amari2016}. Let $\hat\theta:\R_{\geq0}\to\Theta$ denote the parameter estimate and $\Gamma=\Gamma^\top\succ0$ the adaptive gain. For a closed convex cone $K\subset\R^p$ and a symmetric positive-definite matrix $M$, let $\Pi_K^{M}[w]\triangleq\arg\min_{v\in K}\frac12(v-w)^\top M(v-w)$ denote the projection of $w\in\R^p$ onto $K$ in the metric induced by $M$. In Euclidean coordinates, unconstrained gradient adaptation is $\dot{\hat\theta}=-\Gamma\nabla\J_t(\hat\theta)$. On $\Theta$, the corresponding projected flow is
\begin{equation}
    \dot{\hat\theta}=\Pi_{\Tcone(\hat\theta)}^{\Gamma^{-1}}
    [-\Gamma\nabla_\theta\J_t(\hat\theta)],
    \label{eq:gradflow}
\end{equation}
or equivalently $0\in\Gamma^{-1}\dot{\hat\theta}+\nabla\J_t(\hat\theta)+\Ncone(\hat\theta)$. Thus the familiar adaptive-control projection operator is a constrained steepest-descent law: the unconstrained learning velocity is projected onto admissible instantaneous motion. Geometrically, $\Gamma$ specifies a metric or mobility for both descent and projection. If a parameter-dependent metric $M(\theta)\succ0$ is used, the same principle becomes
\begin{equation}
    \dot{\hat\theta}=\Pi_{\Tcone(\hat\theta)}^{M(\hat\theta)}
    [-M(\hat\theta)^{-1}\nabla\J_t(\hat\theta)].
    \label{eq:riemannian}
\end{equation}
Taking $M=\I_{t,T}$, the Fisher information matrix of \eqref{eq:fisher}, yields the natural gradient introduced in information geometry \cite{amari1998}. Gain selection is therefore a choice of geometry in parameter space, with the Fisher metric as one admissible choice. Normalization, covariance adaptation, least-squares gains, and natural-gradient preconditioning can then be compared by the metric with respect to which the estimate descends most steeply. For an inequality representation $\Theta=\{\theta:g(\theta)\leq0\}$ with $g:\R^p\to\R^m$ continuously differentiable and Jacobian $\nabla g\in\R^{m\times p}$, the associated constrained optimization problem may instead be described by the Lagrangian $\mathcal L_t(\theta,\lambda)\triangleq\J_t(\theta)+\lambda^\top g(\theta)$ with multiplier $\lambda\in\R^m_{\geq0}$ and dual gain $\Gamma_\lambda=\Gamma_\lambda^\top\succ0$. Primal descent together with projected dual ascent,
\begin{equation}
 \dot{\hat\theta}=-\Gamma[\nabla\J_t(\hat\theta)+\nabla g(\hat\theta)^\top\lambda],\qquad
 \dot\lambda=\Pi^{\Gamma_\lambda^{-1}}_{T_{\R^m_{\geq0}}(\lambda)}[\Gamma_\lambda g(\hat\theta)],
 \label{eq:primaldual}
\end{equation}
introduces the multiplier force $-\nabla g^\top\lambda$. Projection enforces instantaneous feasibility, whereas primal--dual adaptation addresses constraints through saddle-point dynamics and need not preserve primal feasibility during transients.

The metric interpretation also separates data curvature from curvature imposed by the designer. Adding a regularizer $R(\theta)=\frac{\kappa}{2}\|\theta\|^2$ gives
\begin{equation}
    \nabla^2(\J_t+R)=\nabla^2\J_t+\kappa I\succ0
    \label{eq:regcurv}
\end{equation}
even when $\nabla^2\J_t$ is singular. The optimization problem becomes strongly convex, but the observations have not become more informative. Regularization selects among poorly distinguished parameters; excitation changes what the data distinguish. Three mechanisms therefore confine the estimate, and they differ in origin. Data supply curvature through $\nabla^2\J_t$. Regularization supplies curvature through $R$: the $\sigma$-modification adds $-\sigma\Gamma\hat\theta$ to the update law, which is the gradient descent on $R$ with $\kappa$ equal to the $\sigma$-modification gain, so it bounds the estimate without excitation and shifts the minimizer away from $\theta^\star$, which is the familiar bias; the $e$-modification scales this term by the tracking error so that the confinement vanishes as the error vanishes \cite{ioannou1996,narendra1987}. Hard constraints supply no curvature at all: $\ind$ restricts the admissible set and acts only through the normal-cone reaction. The distinction among data information, soft regularization, and hard feasibility is particularly important for overparameterized or nonlinear adaptive models. Section~\ref{sec:lyap} places the constrained flow \eqref{eq:gradflow} within the Lyapunov analysis of adaptive control.

\section{Energy, Dissipation, and Lyapunov-Based Adaptation}\label{sec:lyap}
The constrained gradient-flow representation \eqref{eq:gradflow} reduces to ordinary gradient flow in the interior; at the boundary, the normal-cone reaction removes the outward component without supplying learning power. For a time-varying information energy $\J_t$ and an interior trajectory,
\begin{equation}
    \frac{d}{dt}\J_t(\hat\theta(t))
    =\partial_t\J_t(\hat\theta)
    -\|\nabla\J_t(\hat\theta)\|_{\Gamma}^2,
    \label{eq:movingpotential}
\end{equation}
where $\|v\|_\Gamma^2\triangleq v^\top\Gamma v$. If $\J$ is time invariant, the first term vanishes and information energy is dissipated monotonically. In online adaptation, however, $\partial_t\J_t$ is the power injected by the time variation of the information energy: parameter motion dissipates information energy while the time variation of $\J_t$ injects power. For the projected flow \eqref{eq:gradflow}, the inclusion $0\in\Gamma^{-1}\dot{\hat\theta}+\nabla\J_t(\hat\theta)+\Ncone(\hat\theta)$ and the complementarity $n^\top\dot{\hat\theta}=0$ satisfied by the active normal $n\triangleq-(\Gamma^{-1}\dot{\hat\theta}+\nabla\J_t(\hat\theta))\in\Ncone(\hat\theta)$ give $\nabla\J_t(\hat\theta)^\top\dot{\hat\theta}=-\|\dot{\hat\theta}\|_{\Gamma^{-1}}^2$, so \eqref{eq:movingpotential} holds with $\|\nabla\J_t\|_\Gamma^2$ replaced by the dissipation $\|\dot{\hat\theta}\|_{\Gamma^{-1}}^2$ of the projected velocity. Moreover, since $\theta^\star\in\Theta$, the definition of $\Ncone(\hat\theta)$ yields
\begin{equation}
    (\theta^\star-\hat\theta)^\top\big(\Gamma^{-1}\dot{\hat\theta}+\nabla\J_t(\hat\theta)\big)\geq0,
    \label{eq:projineq}
\end{equation}
which is the projection inequality of Lyapunov-based adaptive control \cite[Lemma E.1]{krstic1995}: the normal reaction performs no positive work against the true parameter.

The information energy $\J_t$ and the Lyapunov function used for closed-loop stability play different roles. In model-reference and nonlinear adaptive control, one commonly constructs
\begin{equation}
    V(e,\tilde\theta)
    \triangleq V_e(e)+\frac12\tilde\theta^\top\Gamma^{-1}\tilde\theta,
    \label{eq:lyapenergy}
\end{equation}
where $e\in\R^n$ is a tracking error, $V_e:\R^n\to\R_{\geq0}$ is positive definite, and $\tilde\theta\triangleq\theta^\star-\hat\theta$ is the parameter estimation error. The update law is selected so that the cross-terms cancel and $\dot V\leq0$ \cite{narendra1989,sastry1989,ioannou1996}.

The Lyapunov function $V$ and the information energy $\J_t$ therefore play distinct roles, which parallels the classical observation that PE is not required for boundedness or tracking but is central to parameter convergence. Composite and concurrent-learning methods can be viewed as designs that inject explicit information-energy terms into an otherwise Lyapunov-driven adaptive architecture \cite{chowdhary2011,parikh2019,cho2018}. Section~\ref{sec:skew} makes the two roles precise through a metriplectic decomposition of the error dynamics.

\subsection{Metriplectic structure of the error dynamics}\label{sec:skew}
Consider the canonical error model
\begin{equation}
    \dot e=Ae+BY(t)\tilde\theta,\qquad
    V\triangleq\frac12 e^\top Pe+\frac12\tilde\theta^\top\Gamma^{-1}\tilde\theta,
    \label{eq:canonicalerror}
\end{equation}
where $A\in\R^{n\times n}$ is Hurwitz, $B\in\R^{n\times m}$, $Y:\R_{\geq0}\to\R^{m\times p}$ is a bounded regressor, and $P=P^\top\succ0$ satisfies $A^\top P+PA=-Q$ for some $Q=Q^\top\succ0$. The tracking-driven update $\dot{\hat\theta}=\Gamma Y^\top B^\top Pe$ cancels the cross-term $e^\top PBY\tilde\theta$ in $\dot V$. With $z\triangleq[e^\top,\tilde\theta^\top]^\top$ and $\nabla V=[(Pe)^\top,(\Gamma^{-1}\tilde\theta)^\top]^\top$, write $AP^{-1}=S-R_e$ with skew part $S\triangleq\frac12(AP^{-1}-P^{-1}A^\top)$ and symmetric part $R_e\triangleq\frac12P^{-1}QP^{-1}\succ0$, and define
\begin{equation}
    J(t)\triangleq\begin{bmatrix}
        S & BY(t)\Gamma\\
        -\Gamma Y(t)^\top B^\top & 0
    \end{bmatrix},\quad
    R\triangleq\begin{bmatrix}R_e&0\\0&0\end{bmatrix},
    \label{eq:skewadaptive}
\end{equation}
so that $J^\top=-J$ and $R=R^\top\succeq0$. The off-diagonal blocks of $J$ encode the cancellation as a skew interconnection, which is the power-conserving interconnection of port-Hamiltonian systems and is well understood in passivity-based control \cite{vandSchaft2017,ortega2002}. The content of this section is what concurrent learning adds to that structure. Under ideal stored regression data $\{Y_i\}_{i=1}^N$, the concurrent-learning term is
\begin{equation}
    \dot{\tilde\theta}_{\rm CL}=-k_{\rm CL}\Gamma G\tilde\theta,
    \qquad
    G\triangleq\sum_{i=1}^{N}Y_i^\top Y_i\succeq0,
    \label{eq:CLdissipation}
\end{equation}
with gain $k_{\rm CL}>0$. Define the stored-data information energy $\J_{\rm CL}(\tilde\theta)\triangleq\frac{k_{\rm CL}}{2}\tilde\theta^\top G\tilde\theta$ and the symmetric operator $M\triangleq\operatorname{diag}(0,\Gamma)\succeq0$. The following proposition states the resulting structure.
\begin{proposition}
Along the error dynamics \eqref{eq:canonicalerror} with the update $\dot{\hat\theta}=\Gamma Y^\top B^\top Pe+k_{\rm CL}\Gamma\sum_{i=1}^{N}Y_i^\top(y_i-Y_i\hat\theta)$, where $y_i\triangleq Y_i\theta^\star$ denotes the ideal stored output so that the stored-data term equals $k_{\rm CL}\Gamma G\tilde\theta$,
\begin{equation}
    \dot z=[J(t)-R]\nabla V-M\nabla\J_{\rm CL},
    \label{eq:metriplectic}
\end{equation}
where $J$ is skew-symmetric and $R$, $M$ are symmetric positive semidefinite. Along solutions,
\begin{equation}
    \dot V=-\tfrac12 e^\top Qe-k_{\rm CL}\tilde\theta^\top G\tilde\theta,
    \label{eq:CLVdot}
\end{equation}
so $V$ dissipates in every parameter direction if and only if $G\succ0$, and $\dot\J_{\rm CL}=-k_{\rm CL}\tilde\theta^\top G\Gamma Y^\top B^\top Pe-k_{\rm CL}^2\tilde\theta^\top G\Gamma G\tilde\theta$. The GENERIC degeneracy conditions $J\nabla\J_{\rm CL}=0$ and $M\nabla V=0$ fail whenever $BY\Gamma G\tilde\theta\neq0$ and $\tilde\theta\neq0$, respectively.
\end{proposition}
\emph{Proof:} The first block of $[J-R]\nabla V$ is $(S-R_e)Pe+BY\Gamma\Gamma^{-1}\tilde\theta=Ae+BY\tilde\theta=\dot e$. Since $\dot{\tilde\theta}=-\dot{\hat\theta}=-\Gamma Y^\top B^\top Pe-k_{\rm CL}\Gamma G\tilde\theta$, the second block is the second block of $J\nabla V$ minus $\Gamma\nabla\J_{\rm CL}$, which is the second block of $M\nabla\J_{\rm CL}$. Then $\dot V=\nabla V^\top J\nabla V-\nabla V^\top R\nabla V-\nabla V^\top M\nabla\J_{\rm CL}$; the first term vanishes by skew symmetry, $\nabla V^\top R\nabla V=\frac12e^\top Qe$, and $\nabla V^\top M\nabla\J_{\rm CL}=k_{\rm CL}\tilde\theta^\top G\tilde\theta$. The expressions for $\dot\J_{\rm CL}$, $J\nabla\J_{\rm CL}=[k_{\rm CL}BY\Gamma G\tilde\theta;\,0]$, and $M\nabla V=[0;\,\tilde\theta]$ follow by direct computation. \hfill$\square$

In GENERIC thermodynamics, energy and entropy generate the reversible and irreversible parts, and the degeneracy conditions make each part conserve the other generator \cite{grmela1997,ottinger1997}. Proposition~2 places adaptive error dynamics in the same two-generator form, with the Lyapunov function generating the reversible exchange and the tracking dissipation $R$, and the stored-data information energy generating the data-driven irreversible part, but with both degeneracy conditions violated. Their violation is the mechanism of convergence: the irreversible part is designed to decrease $V$, and the reversible exchange drives $\J_{\rm CL}$ through the tracking error. The same $G$ that appears as the excitation Gramian, the accumulated-loss Hessian, and the Fisher information matrix therefore sets the coercivity of the irreversible part, so for the linear parameterization \eqref{eq:regression} full excitation, strong convexity, statistical identifiability, and full parameter-space dissipation are four views of one positive-definiteness condition. With parameter constraints, \eqref{eq:metriplectic} acquires the reaction $[0;\,\Gamma n_\Theta]$ with $n_\Theta\in\Ncone(\hat\theta)$, and $\tilde\theta^\top n_\Theta\leq0$ by the definition of $\Ncone(\hat\theta)$ since $\theta^\star\in\Theta$, so the reaction cannot increase $V$; a primal--dual formulation \eqref{eq:primaldual} instead introduces the multiplier force $\nabla g^\top\lambda$, which for an affine constraint contributes a further skew block between primal and dual states. The reaction is neither generator: feasibility is not excitation. Section~\ref{sec:onsager} derives the projected flow from the speed-gradient and Onsager principles.

\subsection{Speed-gradient and Onsager viewpoints}\label{sec:onsager}
The gradient-flow interpretation also has direct antecedents within adaptive control. Fradkov's speed-gradient method organizes motion in the adjustable parameters using the gradient of the rate of change of a goal functional \cite{fradkov1980,fradkov2020}. Given a goal functional $Q:\R^n\times\R^p\times\R_{\geq0}\to\R$ of the closed-loop state $x\in\R^n$, define its speed along the closed-loop dynamics by
\begin{equation}
    \omega(x,\hat\theta,t)\triangleq\frac{d}{dt}Q(x,\hat\theta,t).
\end{equation}
A canonical speed-gradient update has the form
\begin{equation}
    \dot{\hat\theta}=-\Gamma\nabla_{\hat\theta}\omega.
    \label{eq:speedgradient}
\end{equation}
Speed-gradient adaptation selects parameter motion from the sensitivity of the speed of the goal functional rather than of $Q$ itself. For the error model \eqref{eq:canonicalerror} with $Q(e)=\frac12e^\top Pe$, the speed is $\omega=e^\top P(Ae+BY\tilde\theta)$, so $\nabla_{\hat\theta}\omega=-Y^\top B^\top Pe$ and \eqref{eq:speedgradient} returns the standard update $\dot{\hat\theta}=\Gamma Y^\top B^\top Pe$ of Section~\ref{sec:skew}; more generally, the scheme recovers many adaptation and identification algorithms otherwise obtained by direct Lyapunov arguments \cite{fradkov1980}. The speed-gradient scheme is therefore an early connection between Lyapunov synthesis and gradient-flow adaptation.

Onsager's variational principle for dissipative dynamics \cite{onsager1931a,onsager1931b,doi2011,doi2021,wang2021} provides a second connection. For the information energy $\J_t(\theta)$, introduce the quadratic dissipation potential $\Psi:\R^p\to\R_{\geq0}$,
\begin{equation}
    \Psi(v)\triangleq\frac12 v^\top\Gamma^{-1}v,
\end{equation}
and the instantaneous Rayleighian
\begin{equation}
    \mathcal R_t(v;\theta)
    \triangleq\Psi(v)+\nabla\J_t(\theta)^\top v.
    \label{eq:rayleighian}
\end{equation}
The explicit time dependence $\partial_t\J_t$ does not affect minimization with respect to the velocity $v$. Onsager's principle selects the admissible velocity by
\begin{equation}
    \dot{\hat\theta}
    =\arg\min_{v\in\Tcone(\hat\theta)}\mathcal R_t(v;\hat\theta)
    =\arg\max_{v\in\Tcone(\hat\theta)}
    \{-\nabla\J_t^\top v-\Psi(v)\}.
    \label{eq:onsager}
\end{equation}
The optimality condition is
\begin{equation}
 0\in\Gamma^{-1}\dot{\hat\theta}+\nabla\J_t(\hat\theta)+
 N_{\Tcone(\hat\theta)}(\dot{\hat\theta}),
\end{equation}
which recovers \eqref{eq:gradflow}, since $N_{\Tcone(\hat\theta)}(\dot{\hat\theta})\subset\Ncone(\hat\theta)$. Hence $\Gamma^{-1}$ is a friction metric and projection is the constrained Onsager solution. For a set $\Theta=\{\theta:h_i(\theta)\geq0,\ i=1,\ldots,m\}$ with continuously differentiable $h_i:\R^p\to\R$ (equivalently $g_i=-h_i\leq0$ in the primal--dual convention above) and extended class-$\mathcal K$ functions $\alpha_i$, an alternative viability realization replaces the hard tangent-cone condition by
\begin{equation}
 \nabla h_i(\hat\theta)^\top v+\alpha_i(h_i(\hat\theta))\geq0,
 \label{eq:cbfadapt}
\end{equation}
inside \eqref{eq:onsager}. The resulting QP maximizes learning power subject to barrier constraints: the information objective selects the preferred adaptation velocity while the barrier inequalities render the admissible parameter set forward invariant \cite{ames2017}. Dualizing $g(\theta)\leq0$ instead gives \eqref{eq:primaldual}, where $-\nabla g^\top\lambda$ is a dynamical constraint reaction force. Projection and barrier viability enforce forward invariance; primal--dual flow addresses the same constraints through dynamic multipliers, without generally preserving transient feasibility.

Higher-order tuners already have variational and accelerated-optimization interpretations \cite{gaudio2021,wibisono2016,boffi2021}. Their constrained form gives a useful additional connection. Write an inertial tuner as $\dot{\hat\theta}=\nu$ and
\begin{equation}
 M_\nu\dot\nu=-\nabla\J_t(\hat\theta)-D\nu+u_a,
 \label{eq:hotforce}
\end{equation}
with inertia $M_\nu=M_\nu^\top\succ0$, damping $D=D^\top\succeq0$, and an auxiliary adaptation force $u_a\in\R^p$. For a twice continuously differentiable constraint function $h:\R^p\to\R$ with $h(\hat\theta)\geq0$, the force $u_a$ first appears in $\ddot h$; hence the constraint has relative degree two with respect to $u_a$. Defining $\psi_1\triangleq\dot h+\alpha_1(h)$ and imposing
\begin{equation}
 \dot\psi_1+\alpha_2(\psi_1)\geq0
 \label{eq:hocbfadapt}
\end{equation}
produces an affine constraint on $u_a$ through $\nabla h^\top M_\nu^{-1}u_a$, i.e., the standard high-order CBF mechanism \cite{xiao2022}. Thus first-order tuners admit CBF constraints on adaptation velocity, whereas inertial tuners naturally admit HOCBF constraints on adaptation force. Section~\ref{sec:stoch} lifts the constrained gradient flow and the balance \eqref{eq:movingpotential} to a diffusion on $\Theta$.

\section{Stochastic Adaptation and Free Energy}\label{sec:stoch}
\subsection{From gradient adaptation to Langevin dynamics}\label{sec:langevin}
In this section, assume $\Theta$ has nonempty interior and a $C^2$ boundary when $\partial\Theta\neq\emptyset$, with the regularity and integrability required for the reflected diffusion and density identities. The constrained stochastic counterpart of \eqref{eq:gradflow} is a reflected overdamped Langevin process
\begin{equation}
    d\Theta_t
    =-\Gamma\nabla\J(\Theta_t)dt
    +\sqrt{2\beta^{-1}\Gamma}\,dW_t+dK_t,
    \label{eq:langevin}
\end{equation}
where $\Theta_t\in\Theta$ is the stochastic parameter estimate, $W_t$ is a standard $p$-dimensional Wiener process, $\sqrt{\cdot}$ denotes the symmetric positive-definite square root, $\beta>0$ is the inverse temperature, and $K_t$ is a bounded-variation reflection term supported on $\partial\Theta$ and directed along the inward conormal $\Gamma n_{\rm in}$, with $n_{\rm in}$ the inward unit normal (so along $n_{\rm in}$ itself when $\Gamma=I$). Standard Langevin and Fokker--Planck background is given in \cite{risken1996}. In the zero-temperature limit it reduces to constrained deterministic adaptation; when $\Theta=\R^p$, $K_t=0$ and ordinary Langevin dynamics are recovered.

If $\J(\theta)=\frac12(\theta-\theta^\star)^\top G(\theta-\theta^\star)$ with $G\succ0$, \eqref{eq:langevin} is an Ornstein--Uhlenbeck process on $\R^p$, with a Gaussian invariant density; on a constrained domain it is the reflected counterpart. On a convex $\Theta$ with reflecting boundary, the invariant law is the corresponding truncated Gibbs density
\begin{equation}
    \rho_\infty(\theta)
    \triangleq Z_\Theta^{-1}e^{-\beta\J(\theta)}\mathbf 1_{\Theta}(\theta),\qquad
    Z_\Theta\triangleq\int_\Theta e^{-\beta\J(\theta)}d\theta,
    \label{eq:gibbs}
\end{equation}
where $\mathbf 1_\Theta$ is the indicator of $\Theta$; it reduces to the Gaussian density $\N(\theta^\star,(\beta G)^{-1})$ when $\Theta=\R^p$. It is independent of $\Gamma$, which changes only relaxation rates. When $\Theta=\R^p$, because the drift is affine and there is no reflection term, the mean $m$ and covariance $\Sigma$ satisfy
\begin{equation}
 \dot m=-\Gamma G(m-\theta^\star),\quad
 \dot\Sigma=-\Gamma G\Sigma-\Sigma G\Gamma+2\beta^{-1}\Gamma.
 \label{eq:moments}
\end{equation}
Thus excitation curvature controls stochastic concentration: for $\Gamma=I$, the stationary covariance $(\beta G)^{-1}$ assigns variance $(\beta\lambda_i)^{-1}$ to each eigendirection of $G$ with eigenvalue $\lambda_i>0$, so stronger excitation produces a narrower distribution. Under full excitation the quadratic $\J$ confines the estimate in every direction, whereas under partial excitation ($G\succeq0$ singular) the estimate is confined only transverse to $\ker G$ and diffuses along the feasible null directions. A hard constraint truncates those null directions, and a prior or a regularizer such as $R$ in \eqref{eq:regcurv} prevents the estimate from diffusing far along them, but neither changes the information supplied by the data. Section~\ref{sec:fp} describes the same diffusion at the level of its density.

\subsection{Fokker--Planck dynamics and variational free energy}\label{sec:fp}
In the interior of $\Theta$, the density $\rho_t$ of \eqref{eq:langevin} satisfies
\begin{equation}
    \partial_t\rho
    =\nabla\cdot\big(\Gamma\rho\nabla\J\big)
    +\beta^{-1}\nabla\cdot\big(\Gamma\nabla\rho\big),
    \label{eq:FP}
\end{equation}
with reflecting/no-flux boundary condition $n^\top\Gamma(\rho\nabla\J+\beta^{-1}\nabla\rho)=0$ on $\partial\Theta$, where $n$ is the outward unit normal. Define the constrained free energy
\begin{equation}
    \F[\rho]
    \triangleq\int_{\Theta}\J(\theta)\rho(\theta)d\theta
    +\beta^{-1}\int_{\Theta}\rho(\theta)\log\rho(\theta)d\theta.
    \label{eq:freeenergy}
\end{equation}
The first term is the expected information energy; the second is negative entropy scaled by temperature. Let $\mathcal P_2(\Theta)$ denote the probability densities on $\Theta$ with finite second moment, let $W_2$ denote the Wasserstein-2 distance on $\mathcal P_2(\Theta)$ \cite{villani2009}, and let $\operatorname{grad}_{W_2}\F[\rho]\triangleq-\nabla\cdot\big(\rho\nabla\tfrac{\delta\F}{\delta\rho}\big)$ denote the Wasserstein gradient of $\F$, where $\tfrac{\delta\F}{\delta\rho}=\J+\beta^{-1}(\log\rho+1)$ is the first variation \cite{ambrosio2008}. Jordan, Kinderlehrer, and Otto showed that for $\Gamma=I$ the Fokker--Planck equation \eqref{eq:FP} is the steepest descent of $\F$ in this metric \cite{jordan1998},
\begin{equation}
    \partial_t\rho_t=-\operatorname{grad}_{W_2}\F[\rho_t],
    \label{eq:wgf}
\end{equation}
and a general mobility $\Gamma$ corresponds to the transport cost weighted by $\Gamma^{-1}$. Deterministic adaptation \eqref{eq:gradflow} and \eqref{eq:wgf} are therefore the same dissipative flow generated by information energy plus an admissibility constraint, the first on a feasible parameter and the second on a probability law supported on $\Theta$.

For the Gibbs density $\rho_\infty$ in \eqref{eq:gibbs} and the relative entropy $\KL(\rho\|\rho_\infty)\triangleq\int_\Theta\rho\log(\rho/\rho_\infty)d\theta$, expanding the logarithm in \eqref{eq:freeenergy} gives
\begin{equation}
    \F[\rho]=\beta^{-1}\KL(\rho\|\rho_\infty)-\beta^{-1}\log Z_\Theta.
    \label{eq:KLfree}
\end{equation}
Thus free-energy dissipation is simultaneously relaxation toward a Gibbs law and reduction of a relative-information measure.

For a moving potential $\J_t$, the deterministic identity \eqref{eq:movingpotential} lifts to a nonequilibrium balance. Under standard smoothness and decay assumptions,
\begin{equation}
 \frac{d}{dt}\F_t[\rho_t]
 =\int_\Theta \partial_t\J_t(\theta)\rho_t(\theta)d\theta
 -\mathcal D_t[\rho_t],
 \label{eq:firstlaw}
\end{equation}
where
\begin{equation}
 \mathcal D_t[\rho]\triangleq\int_\Theta \rho\,\|\nabla(\J_t+\beta^{-1}\log\rho)\|_\Gamma^2d\theta\ge0.
\end{equation}
Equation \eqref{eq:firstlaw} follows by differentiating \eqref{eq:freeenergy}, substituting \eqref{eq:FP}, and integrating by parts, the no-flux condition removing boundary power; it is the first-law balance of stochastic thermodynamics \cite{seifert2012}: the free-energy rate equals the work done by the time variation of $\J_t$ minus the dissipation. In the zero-temperature Dirac-mass limit it reduces formally to \eqref{eq:movingpotential}. Hence the same obstruction appears in deterministic and stochastic adaptation: the information energy need not decrease instantaneously because the time variation of $\J_t$ can inject power. Section~\ref{sec:mi} replaces local curvature by finite information gain.

\subsection{Mutual information and experiment design}\label{sec:mi}
The Fisher interpretation can be strengthened from local curvature to finite information gain. For sampled observations $y_k=\phi_k^\top\theta^\star+w_k$, $k=1,\ldots,n$, with independent $w_k\sim\N(0,\sigma^2)$ and a Gaussian prior $\theta^\star\sim\N(\theta_0,P_0)$, let $S_n\triangleq\sum_{k=1}^n\phi_k\phi_k^\top$, let $P_n$ denote the posterior covariance, and let $\operatorname{MI}(\theta^\star;y_{1:n})$ denote the mutual information between $\theta^\star$ and $y_{1:n}\triangleq(y_1,\ldots,y_n)$ \cite{cover2006}. Then
\begin{equation}
\begin{aligned}
 P_n^{-1}&=P_0^{-1}+\sigma^{-2}S_n,\\
 \operatorname{MI}(\theta^\star;y_{1:n})&=\frac12\log\det(I+\sigma^{-2}P_0S_n).
\end{aligned}
\label{eq:mutualinfo}
\end{equation}
Thus an unexcited direction acquires no data information, while full-rank accumulated excitation contracts posterior uncertainty in every identifiable direction. If prior knowledge restricts $\theta^\star$ to $\Theta$, the posterior is truncated: the reduction in uncertainty volume is prior information, not excitation. Maximizing $\log\det S_n$ is the familiar D-optimal experiment-design criterion \cite{goodwinpayne1977,pukelsheim2006}. From this viewpoint, controlled excitation shapes an information volume, making the classical regulation--excitation tension a control-versus-information-design tradeoff. Section~\ref{sec:unified} extends the dictionary to nonlinear parameterizations.

\section{Extension to Nonlinear Parameterizations}\label{sec:unified}
The machinery of the preceding sections also applies when the model is nonlinear in its parameters. The main example in adaptive control is the Lyapunov-based deep neural network (Lb-DNN) architecture. A twice continuously differentiable DNN $\Phi:\Omega\times\R^p\to\R^q$ replaces the regressor, its weight vector $\theta\in\R^p$ collects all layers, and the weights of every layer are adapted online by update laws that are derived from a Lyapunov analysis rather than trained offline \cite{patil2022,patil2023composite}. The analysis rests on a first-order expansion of the network in its weights, $\Phi(x,\theta^\star)=\Phi(x,\hat\theta)+\Phi'(x,\hat\theta)\tilde\theta+\mathcal O(\|\tilde\theta\|^2)$, where $\Phi'\triangleq\partial\Phi/\partial\theta\in\R^{q\times p}$ is the Jacobian. The Jacobian therefore plays the role of $\phi^\top$ in \eqref{eq:regression}, and the remainder enters the closed loop as a bounded disturbance.

Every object in the dictionary has a counterpart in this setting. The accumulated loss is \eqref{eq:generalJ} with an input-space measure, and its Hessian has three parts: the Jacobian Gramian $\int_\Omega\Phi'^\top\Phi'\,dx$, the regularizer $\kappa I$ of \eqref{eq:regcurv}, and a residual term proportional to the approximation error, which is bounded on a compact set. Data curvature and regularization contribute to confinement, while residual curvature need not be positive. Uniqueness of the regularized loss minimizer over a compact convex weight set follows if the full Hessian is positive definite throughout the set. A sufficient condition combines a uniformly positive-definite Jacobian Gramian with regularization chosen to dominate residual curvature; stored-input rank tests evaluate the Gramian component of this condition \cite{hart2026}. Stored-data prediction gradients provide the concurrent-learning contribution to \eqref{eq:gradflow}, and the smallest eigenvalue of the stacked Jacobian Gramian controls local linearized dissipation, together with the adaptive gains and remainder bounds \cite{hart2025,hart2026}. Continuous weight symmetries can make this Gramian singular, while discrete symmetries can yield equivalent parameters without a local rank defect. Regularization can restore uniqueness of the loss minimizer when the full-Hessian condition holds, at a cost in approximation accuracy \cite{hart2026}.

The open questions are the ones the dictionary makes visible. Which stored inputs maximize the information volume of the Jacobian Gramian is the D-optimal design problem of Section~\ref{sec:mi} in input space. How much regularization to keep once the data supply sufficient curvature is a data-dependent confinement problem. Structured weight sets call for the barrier realizations of Section~\ref{sec:onsager}, and flat weight directions call for the Langevin lift \eqref{eq:langevin}. The deterministic additions keep the closed-loop certificate \eqref{eq:lyapenergy} whenever the added term has the form $-M\nabla\J$ of \eqref{eq:metriplectic} for a convex $\J$ minimized at $\tilde\theta=0$, and the Langevin lift keeps the free-energy balance \eqref{eq:firstlaw} in place of the pointwise certificate. Section~\ref{sec:conclusion} concludes.

\section{Conclusion}\label{sec:conclusion}
A temporally accumulated prediction loss provides a common language for adaptive learning: its curvature is the excitation Gramian and, under Gaussian observations, Fisher information; memory retains curvature, concurrent learning supplies parameter-space dissipation, and constraints act through admissible geometry rather than data information. Projected, primal--dual, barrier-constrained, Langevin, and free-energy formulations are different realizations of the same information--dissipation--constraint structure. The perspective makes the learning and identifiability dynamics of Lyapunov-based adaptive control explicit, and the metriplectic decomposition shows how any convex information energy can be added to a certified design.

\balance
\bibliographystyle{IEEEtran}
\bibliography{references}

\end{document}